\documentclass[11pt, a4paper]{article}
\usepackage{amsthm, amsmath, amssymb, graphicx, enumerate}

\usepackage[margin=2.5cm]{geometry}
\usepackage{setspace}  

\newcommand{\OPT}{\textrm{OPT}}

\newcommand{\NP}{\mathsf{NP}}
\newcommand{\APX}{\mathsf{APX}}
\newcommand{\stack}{Stackelberg network pricing}
\newcommand{\toll}{toll setting}
\newcommand{\river}{river tariff pricing}
\newcommand{\envyfree}{profit-maximizing envy-free pricing}
\newcommand{\tar}{d}                      

\newtheorem{theorem}{Theorem}
\newtheorem{lemma}{Lemma}
\newtheorem{claim}{Claim}

\title{Stackelberg Network Pricing is Hard to Approximate}

\author{
Gwena\"el Joret\thanks{Universit\'e Libre de Bruxelles (ULB), Facult\'e des Sciences, 
D\'epartement d'Informatique, c.p.~212,  B-1050 Brussels, Belgium, {\tt gjoret@ulb.ac.be}. 
Postdoctoral Researcher of the {\em Fonds
National de la Recherche Scientifique (F.R.S.--FNRS)}.}
}

\date{}

\begin{document}
\maketitle
\begin{abstract}
In the {\stack} problem, one has
to assign tariffs to a certain subset of the arcs of a given transportation network. 
The aim is to maximize the amount paid by the user of the network, knowing that the user 
will take a shortest $st$-path once the tariffs are fixed.
Roch, Savard, and Marcotte ({\it Networks}, Vol. 46(1), 57--67, 2005)
proved that this problem is $\NP$-hard, 
and gave an $O(\log m)$-approximation algorithm, where
$m$ denote the number of arcs to be priced.
In this note, we show that the problem is also $\APX$-hard.
\end{abstract}


\section{Introduction}

We consider a network pricing problem involving two non-cooperative players,
a {\em leader} and a {\em follower}. 
The leader owns a subset of the arcs of a given transportation
network. She has to set tariffs on these arcs, knowing that the follower will compute a shortest
$st$-path once the tariffs are fixed. The goal is to maximize
the revenue of the leader, which depends on the path chosen by the follower.

This problem is known as {\it \stack}; it is formally described as follows: 

\begin{figure}[ht]
\begin{minipage}{0.14\textwidth}
$ $
\end{minipage}
\begin{minipage}{0.85\textwidth}
\begin{itemize}
\item[\bf{INSTANCE:}]
\begin{itemize}
\item a directed graph $D=(V, A)$
\item a cost function $c: A \to \mathbb{R}_{+}$
\item a pair $(s, t)$ of distinct nodes $s, t \in V$   
\item a subset $T \subseteq A$ of {\it tariff arcs}
\end{itemize}
\item[\bf{SOLUTION:}]
\begin{itemize}
\item an assignment $\tar: T \to \mathbb{R}_{+}$ of tariffs to the arcs in $T$ 
\item an $st$-path $P$ of $D$ minimizing its total cost 
$\sum_{a\in P} c(a) + \sum_{a\in P \cap T} \tar(a)$
\end{itemize}
\item[\bf{OBJECTIVE:}]
\begin{itemize}
\item maximize the revenue $\sum_{a\in P \cap T} \tar(a)$
\end{itemize}
\end{itemize}
\end{minipage}
\end{figure}

Before going further, 
let us make two remarks on the above formulation: First, it is usually assumed that
there exists an $st$-path in $D$ that uses only arcs of $A - T$, since otherwise the optimum
is unbounded. Second, once the tariffs are fixed,
we can easily choose among all $st$-paths of minimum (total) cost one path $P$ that maximizes
the revenue. In other words, we assume
that the follower always makes the best choice for the leader.
This is a standard assumption, justified by the fact 
that decreasing the prices of every arc in $P \cap T$ by an arbitrarily small amount
ensures that $P$ is the unique $st$-path of minimum cost.

The {\stack} problem has recently been studied by Roch, Savard, and Marcotte~\cite{RSM05},
motivated by applications in transportation and telecommunications.
They proved that the problem is $\NP$-hard and described a polynomial-time algorithm
approximating the optimum within a ratio of $\frac{1}{2} \log_{2} |T| + 1$.

The purpose of this note is to show that the {\stack} problem is also $\APX$-hard:

\begin{theorem}
\label{th-APX}
For some $\epsilon > 0$, it is $\NP$-hard to approximate
{\stack} within a ratio of $1 - \epsilon$.
\end{theorem}

We conclude this introduction by mentioning other related works.
A generalization of {\stack}, the {\it {\toll} problem},
has been considered by Labb\'e, Marcotte, and Savard~\cite{LMS98}.
It involves multiple, weighted followers: the $i$th follower 
computes a shortest $s_{i}t_{i}$-path in the network once the tariffs
are fixed, and has a demand of $d_{i}$.
The objective is then to maximize the sum of the revenues obtained 
from each individual follower, weighted by their respective demands.
As observed by Roch {\it et al.}~\cite{RSM05}, 
the approximation algorithm mentioned above for {\stack}
directly gives a $O(k\log m)$-approximation algorithm in the case of unit demands, where
$k$ denotes the number of followers. 

A special case of the {\toll} problem occurs when every $s_{i}t_{i}$-path in $D$
uses at most one arc in $T$, and is known under the name {\it {\river}}. 
It has been considered by 
Bouhtou, Grigoriev, van Hoesel, van der Kraaij, Spieksma, and Uetz~\cite{BGvHvdKSU07}, 
and is closely related to the {\envyfree} problem
studied by Guruswami, Hartline, Karlin, Kempe, Kenyon, and McSherry~\cite{GHKKKS05}. 
Among others, both sets of authors proved 
(independently) that the {\river} problem is $\APX$-hard, even in the case of unit demands.
Briest, Hoefer, and Krysta~\cite{BHK08} later derived stronger inapproximability results,
relying on a recent inapproximability result of 
Demaine, Feige, Hajiaghayi, and Salavatipour~\cite{DFHStoappear} 
for {\envyfree}.
We note that each of these inapproximability results uses in a crucial way the fact that  
there are multiple followers. 
 
Finally, we mention that some other combinatorial optimization problems 
similar to {\stack} have been considered recently. 
This includes pricing edges of an undirected graph 
knowing that the follower will compute a minimum spanning tree~\cite{CDFJLNW07}, and
pricing vertices of a bipartite undirected graph when the follower buys 
a minimum cost vertex cover~\cite{BHK08}.
Other kinds of Stackelberg games in networks have been studied by 
Cole, Dodis, and Roughgarden~\cite{CDR03}, Roughgarden~\cite{R04}, and
Swamy~\cite{S07}.

\section{The Proof}

In order to prove Theorem~\ref{th-APX}, we need the following lemma 
on bounded-degree graphs. Here and throughout the text,
a {\sl linear ordering} of the vertices of a graph $G=(V,E)$ 
is a bijective mapping $\ell:V \rightarrow \{1, \dots, |V|\}$.

\begin{lemma}
\label{lem-lo}
Let $G=(V,E)$ be an undirected graph with maximum degree $\Delta \geq 1$ and
$n\geq c_{\Delta}$ vertices, where  $c_{\Delta} := 4 \Delta(\Delta +1)$.
Then a linear ordering $\ell$ of $V$ with 
$$
|\ell(u) - \ell(v)| \ge \frac{n}{c_{\Delta}}
$$
for every edge $uv\in E$ can be found in polynomial time.
\end{lemma}

A (proper, vertex) coloring of a graph $G$
is {\sl equitable} if every two color classes differ in size by at most 1. 
The proof of Lemma~\ref{lem-lo} relies on
the following result on equitable colorings.

\begin{theorem}[Hajnal and Szemer{\'e}di~\cite{HS70}]
\label{th-delta-plus-1}
Every graph $G$ with maximum degree $\Delta$ can be equitably colored
with $\Delta + 1$ colors.
\end{theorem}

Kierstead and Kostochka~\cite{KK08} recently obtained a short proof
of Theorem~\ref{th-delta-plus-1}. Their proof yields also a polynomial-time algorithm
finding such a coloring. 

\begin{proof}[Proof of Lemma~\ref{lem-lo}]
The lemma is easily seen to hold if $\Delta=1$, hence we assume $\Delta \geq 2$.
Using the algorithmic version of Theorem~\ref{th-delta-plus-1} given by 
Kierstead and Kostochka~\cite{KK08},
we first find in polynomial time an equitable coloring $S_1, \dots, S_{\Delta + 1}$ of the vertices
of $G$. 
Let 
$$
s:=\left\lfloor \frac{n}{2\Delta(\Delta + 1)}\right\rfloor.
$$
We have
$$
s \ge \frac{n}{2\Delta(\Delta + 1)} - 1 \geq 
\frac{n}{2\Delta(\Delta + 1)} - \frac{n}{4\Delta(\Delta + 1)} = \frac{n}{c_{\Delta}}.
$$
We show that a linear ordering $\ell$ of the vertices of $G$ with
$|\ell(u) - \ell(v)| \ge s$ for every edge $uv\in E$ can be
found in polynomial time, which implies the claim.

Observe that
\begin{equation}
\label{eq-size-Si}
|S_{i}| > s(\Delta + 1) \quad \quad \textrm{ for every } i\in \{1,\dots, \Delta + 1\}.
\end{equation}
Indeed, using $\Delta\geq 2$ and $n \geq 4\Delta(\Delta + 1)$,
Inequality~\eqref{eq-size-Si} can be derived as follows:
$$
|S_{i}| > \frac{n}{\Delta + 1} - \Delta  
= \frac{4n - 4\Delta(\Delta + 1)}{4(\Delta + 1)}  
\geq \frac{3}{4} \cdot \frac{n}{(\Delta + 1)} 
\geq  \frac{\Delta + 1}{2\Delta} \cdot \frac{n}{(\Delta + 1)} 
\geq s(\Delta + 1).
$$

We partition each set $S_{i}$ into three subsets $S_{i}^{1}, S_{i}^{2}, S_{i}^{3}$
as follows: First, the partition of $S_{1}$ is chosen arbitrarily, ensuring only
$|S_{1}^{1}|=|S_{1}^{3}|=s$. Then, for $i=2,3,\dots, \Delta +1$, define $S_{i}^{1}$ as a subset of
$s$ vertices of $S_{i}$ having no neighbor in $S_{i-1}^{3}$. Such a subset
always exists because, by~\eqref{eq-size-Si}, there are most  
$|S_{i-1}^{3}|\cdot \Delta = s \Delta < |S_{i}| - s$ 
vertices in $S_{i}$ with a neighbor in $S_{i-1}^{3}$. Let then $S_{i}^{3}$ be any
subset of $S_{i} - S_{i}^{1}$ with cardinality $s$, and 
set finally $S_{i}^{2} := S_{i} - (S_{i}^{1} \cup S_{i}^{3})$.

Consider the partial order $\prec$ on $V$ where, for $u,v\in V$, we have $u \prec v$ 
if $u\in S_{i}^{p}$, $v\in S_{j}^{q}$ with $i<j$, or with $i=j$ and  $p < q$. 
Define then $\ell$ as any linear ordering of $V$ 
compatible with $\prec$. 

By construction, the set $S_{i-1}^{3} \cup S_{i}^{1}$ is
a stable set for every $i\in \{2, \dots, \Delta +1\}$. It follows  $|\ell(u) - \ell(v)| \geq s$
for every edge $uv \in E$.
Since the linear ordering $\ell$ can clearly be computed 
in polynomial time, this completes the proof of the lemma.
\end{proof}

We turn now to the proof of Theorem~\ref{th-APX}. We note that the gadgets used in the reduction
are essentially the same as the one used by Roch {\it et al}.~\cite{RSM05} in their
$\NP$-hardness proof.

\begin{proof}[Proof of Theorem~\ref{th-APX}]
A 3SAT-5 formula is a CNF formula in which every clause contains 
exactly three literals, every variable appears in exactly five
clauses, and a variable does not appear in a clause more than once. 
Such a formula is said to be {\sl  $\delta$-satisfiable} if at most a $\delta$-fraction
of its clauses are satisfiable simultaneously. 
Our reduction is from the problem of distinguishing
between satisfiable 3SAT-5 formulae and those which are $\delta$-satisfiable. It is known that
this problem is $\NP$-hard for some constant $\delta$ with $0 < \delta < 1$; see Feige~\cite{F98}.

Suppose thus that we are given a 3SAT-5 formula $\varphi$ with $n$ clauses which is either
satisfiable or $\delta$-satisfiable. 
Two literals of  $\varphi$ are {\sl opposite} if one is positive, the other negative, and they both
correspond to the same variable.
Let $G_{\varphi}$ be the (simple, undirected) graph having 
one vertex per clause of $\varphi$, and where two distinct vertices are adjacent 
if there exists two opposite literals in the union of the corresponding two clauses.
Notice that $G_{\varphi}$ has maximum degree  $\Delta \le 12$.
 
Using Lemma~\ref{lem-lo}, we obtain in polynomial time a linear ordering $\ell$ of 
the vertices of $G_{\varphi}$ such that
$$
|\ell(u) - \ell(v)| \geq \frac{n}{4\Delta(\Delta + 1)} \geq \frac{n}{624}
$$
for every edge $uv \in E$.
Denote by $C_1, \dots, C_n$ the clauses of $\varphi$, in the order given by $\ell$.
Denote also by $C_{i}^{1}, C_{i}^{2}, C_{i}^{3}$ the three literals 
of $C_i$, for every  $i \in \{1, \dots, n\}$.

We define an instance of the {\stack} problem as follows.
Each clause $C_{i}$ has a corresponding {\sl clause-gadget}, 
described in Figure~\ref{fig-gadget}. 
\begin{figure}
\centering
\includegraphics[width=0.3\textwidth]{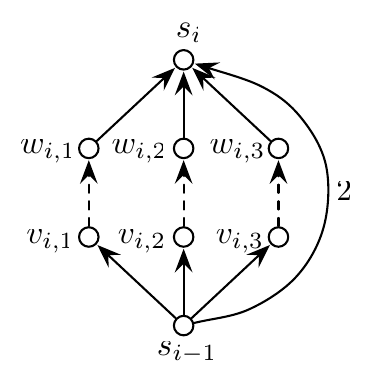}
\caption{Gadget for clause $C_i$. Tariff arcs are represented with dashed 
lines and only non-zero fixed costs are indicated.}
\label{fig-gadget}
\end{figure}
Take first the union 
of all these clause-gadgets (notice that node $s_{i}$
appears in the gadgets of both $C_{i}$ and $C_{i+1}$, for $i \in \{1, \dots, n-1\}$).
Then, for every $i,j \in \{2, \dots, n\}$ with $i<j$ and every $p,q \in \{1,2,3\}$,
add the arc $(w_{i,p}, v_{j,q})$ 
with cost $j - i - 1$ if literals $C_{i}^{p}$ and $C_{j}^{q}$ are opposite.
The latter arcs are said to be {\sl jump arcs}. This defines the directed graph $D=(V,A)$
and the cost function $c(\cdot)$. The set of tariff arcs is 
$$
T:= \{ (v_{i,j}, w_{i,j}): 1 \leq i \leq n, 1\leq j \leq 3\}, 
$$
and the origin-destination pair of the follower is $(s, t)$, with $s:=s_{0}$ and $t:=s_{n}$.

If $n < 1248$, we simply use brute force to decide whether $\varphi$ is satisfiable
or only $\delta$-satisfiable. Hence, we may assume $n \geq 1248$, and thus  
\begin{equation}
\label{eq-jump-arc}
j - i \geq\frac{n}{624} \geq 1 + \frac{n}{1248}
\end{equation}
for every jump arc $(w_{i,p},v_{j,q})$.
Let
$$
\lambda := \max \left\{\delta, 1- \frac{1}{2496}\right\},
$$
and denote by $\OPT$ the maximum revenue achievable on this instance of {\stack}. 

\begin{claim}
\label{claim-gap}
\begin{enumerate}[(a)]
The following holds:
\item \label{enum-gap-sat} if $\varphi$ is satisfiable then $\OPT=2n$;
\item \label{enum-gap-delta-sat} if $\varphi$ is $\delta$-satisfiable then $\OPT \le \lambda \cdot 2n$.
\end{enumerate}
\end{claim}
\begin{proof}

Suppose first that $\varphi$ is satisfiable and 
consider a truth assignment of the variables
satisfying all clauses of $\varphi$.
For every $i\in \{1, \dots, n\}$ and $p\in \{1,2,3\}$, 
set the tariff of arc $(v_{i,p}, w_{i,p})$ to 2 if the literal $C_{i}^{p}$
is true in the truth assignment, to $2n+1$ otherwise. 
Observe that, with these tariffs, any $st$-path that includes a jump arc has (total) cost at least
$2n + 1$. Also, the cost of every $st$-path is at least $2n$, and there exists
one such path with cost exactly $2n$ that uses one tariff arc per clause-gadget.
Hence, $\OPT \geq 2n$ in this case. On the other hand, $\OPT \leq 2n$ always holds since there
exists an $st$-path with cost $2n$ in $(V, A - T)$, that is, which avoids all tariff arcs.
This proves part~\eqref{enum-gap-sat} of the claim.

Assume now that $\varphi$ is 
$\delta$-satisfiable and let $\tar(\cdot)$ be an optimal assignment of tariffs to arcs in $T$.
Let $P$ be any $st$-path of minimum (total) cost giving a
revenue of $\OPT$, and denote by $z$ its cost. 
(Thus $z = \sum_{a\in P}c(a) + \sum_{a\in P\cap T}\tar(a) \leq 2n$.)
If the path $P$ includes a jump arc $(w_{i,p},v_{j,q})$, 
then using~\eqref{eq-jump-arc} we obtain
$$
\OPT \le z - (j-i-1) \le 2n - (j-i-1) \leq 2n - \frac{n}{1248} \le \lambda \cdot 2n.
$$
Hence, without loss of generality $P$ includes no jump arc. 
It follows $\tar(a) \leq 2$
for every tariff arc $a\in P \cap T$, because of the arcs $(s_{i-1}, s_{i})$ with fixed cost 2. 

Suppose $P$ includes two tariff arcs corresponding to
opposite literals, say arcs $(v_{i,p}, w_{i,p})$ and $(v_{j,q}, w_{j,q})$ with $i < j$. 
Then the revenue 
$$
\sum_{a\in P' \cap T} \tar(a)
$$
given by the subpath $P'$ of $P$ going from $w_{i,p}$ 
to $v_{j,q}$ is at most $j - i - 1$. This is because there exists a jump arc $(w_{i,p},v_{j,q})$
in $D$, with fixed cost $j - i - 1$. Hence, we deduce
$$
\OPT = \sum_{a\in (P- P') \cap T} \tar(a) + \sum_{a\in P' \cap T} \tar(a)
 \le 2\big(n - (j-i-1)\big) + (j-i-1)\le \lambda \cdot 2n,
$$
using again~\eqref{eq-jump-arc}. 
We may thus assume that $P$ does not contain two tariff arcs corresponding to
opposite literals. 

Now, the set $P \cap T$ of tariff arcs included in $P$ 
directly defines a truth assignment that satisfies at least $|P \cap T|$ 
clauses of $\varphi$ (variables not appearing in $P \cap T$ are set arbitrarily).
Since $\varphi$ is only $\delta$-satisfiable, we have $|P \cap T| \le \delta$, and thus
$$
\OPT \le 2|P \cap T| \le \delta \cdot 2n \le \lambda \cdot 2n.
$$
Part \eqref{enum-gap-delta-sat} of the claim follows.
\end{proof}

By Claim~\ref{claim-gap}, any polynomial-time algorithm approximating {\stack} within a ratio 
strictly better than  $\lambda$ could be used to decide, 
in polynomial time, whether $\varphi$ is satisfiable or 
$\delta$-satisfiable. This completes the proof of Theorem~\ref{th-APX}.
\end{proof}

\paragraph{Acknowledgments.}
The author wishes to thank Adrian Vetta for useful discussions,
and Jean Cardinal and Samuel Fiorini for their comments on an earlier version of this note. 
This work was supported by the {\em Actions de Recherche Concert\'ees (ARC)\,} fund of 
the {\em Communaut\'e fran\c{c}aise de Belgique}.

\bibliography{maxtoll}

\begin{thebibliography}{10}

\bibitem{BGvHvdKSU07}
M.~Bouhtou, A.~Grigoriev, S.~van Hoesel, A.~F. van~der Kraaij, F.~C.~R.
  Spieksma, and M.~Uetz.
\newblock Pricing bridges to cross a river.
\newblock {\em Naval Res. Logist.}, 54(4):411--420, 2007.

\bibitem{BHK08}
P.~Briest, M.~Hoefer, and P.~Krysta.
\newblock Stackelberg network pricing games.
\newblock In {\em Proc. 25th International Symposium on Theoretical Aspects of
  Computer Science ({STACS})}, pages 133--142, 2008.

\bibitem{CDFJLNW07}
J.~Cardinal, E.~D. Demaine, S.~Fiorini, G.~Joret, S.~Langerman, I.~Newman, and
  O.~Weimann.
\newblock The stackelberg minimum spanning tree game.
\newblock In {\em Proc. 10th international Workshop on Algorithms and Data
  Structures (WADS)}, volume 4619 of {\em Lecture Notes in Computer Science},
  pages 64--76. Springer-Verlag, 2007.

\bibitem{CDR03}
R.~Cole, Y.~Dodis, and T.~Roughgarden.
\newblock Pricing network edges for heterogeneous selfish users.
\newblock In {\em Proc. 35th {A}nnual {ACM} {S}ymposium on {T}heory of
  {C}omputing ({STOC})}, pages 521--530, New York, 2003. ACM.

\bibitem{DFHStoappear}
E.~D. Demaine, U.~Feige, M.~Hajiaghayi, and M.~R. Salavatipour.
\newblock Combination can be hard: Approximability of the unique coverage
  problem.
\newblock {\em {SIAM} J. Comput.}, to appear.

\bibitem{F98}
U.~Feige.
\newblock A threshold of {$\ln n$} for approximating set cover.
\newblock {\em J. ACM}, 45(4):634--652, 1998.

\bibitem{GHKKKS05}
V.~Guruswami, J.~Hartline, A.~Karlin, D.~Kempe, C.~Kenyon, and F.~McSherry.
\newblock On profit-maximizing envy-free pricing.
\newblock In {\em Proc. 16th Annual {ACM}-{SIAM} {S}ymposium on {D}iscrete
  {A}lgorithms ({SODA})}, pages 1164--1173, New-York, 2005. ACM.

\bibitem{HS70}
A.~Hajnal and E.~Szemer{\'e}di.
\newblock Proof of a conjecture of {P}. {E}rd{\H o}s.
\newblock In {\em Combinatorial theory and its applications, II (Proc. Colloq.,
  Balatonf\"ured, 1969)}, pages 601--623. North-Holland, Amsterdam, 1970.

\bibitem{KK08}
H.~A. Kierstead and A.~V. Kostochka.
\newblock A short proof of the {H}ajnal-{S}zemer\'edi theorem on equitable
  colouring.
\newblock {\em Combin. Probab. Comput.}, 17(2):265--270, 2008.

\bibitem{LMS98}
M.~Labb\'e, P.~Marcotte, and G.~Savard.
\newblock A bilevel model of taxation and its application to optimal highway
  pricing.
\newblock {\em Management Science}, 44(12):1608--1622, 1998.

\bibitem{RSM05}
S.~Roch, G.~Savard, and P.~Marcotte.
\newblock An approximation algorithm for {S}tackelberg network pricing.
\newblock {\em Networks}, 46(1):57--67, 2005.

\bibitem{R04}
T.~Roughgarden.
\newblock Stackelberg scheduling strategies.
\newblock {\em SIAM J. Comput.}, 33(2):332--350 (electronic), 2004.

\bibitem{S07}
C.~Swamy.
\newblock The effectiveness of stackelberg strategies and tolls for network
  congestion games.
\newblock In {\em Proc. 18th annual {ACM}-{SIAM} {S}ymposium on {D}iscrete
  {A}lgorithms ({SODA})}, pages 1133--1142, Philadelphia, PA, USA, 2007.
  Society for Industrial and Applied Mathematics.

\end{thebibliography}
\bibliographystyle{plain}

\end{document}